\documentclass[twoside,leqno,twocolumn,10pt]{article}
\usepackage[algo2e,tworuled,linesnumbered]{algorithm2e}
\usepackage{ltexpprt}
\usepackage{booktabs} 

\usepackage{cite}
\usepackage{soul}
\usepackage{amsmath,amssymb,amsfonts}
\usepackage{graphicx}
\usepackage{textcomp}
\usepackage{xcolor}
\usepackage{bbold}
\usepackage{pgf,tikz}
\usepackage{comment}
\usetikzlibrary{hobby}
\usepackage{tikz-cd}
\usetikzlibrary{arrows}
\usetikzlibrary{cd}
\usepackage{enumitem}
\usepackage{comment}
\usepackage{wrapfig}
\usepackage{mathtools}

\usepackage[normalem]{ulem} 

\definecolor{trueblue}{RGB}{0, 115, 197} 

\def\b#1{\boldsymbol{#1}} 
\def\bb#1{\underline{\boldsymbol{#1}}} 
\def\A{\mathcal A} 
\def\L{\mathcal L}
\def\II{\mathcal I}

\def\RR{\mathbb R} 
\def\RRR{\underline{\mathbb R}} 
\def\F{F} 
\def\C{\mathcal C}
\def\matrix{M_{\b \alpha}} 

\newtheorem{definition}{Definition}

\title{Multi-Dimensional, Multilayer, Nonlinear and Dynamic HITS
\thanks{The work of FA was supported by EP/M00158X/1 from the EPSRC/RCUK Digital Economy Programme. The work of FT was supported by European Union's Horizon 2020 research and innovation programme under the Marie Sk\l{}odowska-Curie individual fellowship  ``MAGNET'' No 744014. }
}

\author{Francesca Arrigo\thanks{{\tt francesca.arrigo@stath.ac.uk} University of Strathclyde. } \\
\and
Francesco Tudisco\thanks{{\tt f.tudisco@strath.ac.uk} University of Strathclyde.}}

\begin{document}

\maketitle

\begin{abstract}
We introduce a ranking model for temporal multi-dimensional weighted and directed networks based on the Perron eigenvector of a multi-homogeneous order-preserving map. The model extends to the temporal multilayer setting the HITS algorithm and defines five centrality vectors: two for the nodes, two for the layers, and one for the temporal stamps. 
Nonlinearity is introduced in the standard HITS model in order to guarantee existence and uniqueness of these centrality vectors for any network, without any requirement on its  connectivity structure. We introduce a globally convergent power iteration like algorithm for the computation of the centrality vectors. Numerical experiments on real-world networks are performed in order to assess the effectiveness of the  proposed model and showcase the performance of the accompanying algorithm. 
\end{abstract}



\section{Introduction}\label{sec:introduction}
Locating and evaluating relevant components is a central task in data analysis and information retrieval. One of the most successful approaches creates a network of relations from the data, thus translating the original problem into that of quantifying the importance of nodes in a network. This problem can then be tackled using measures of importance for nodes (that do not depend on their labelling) which we will refer to as \textit{centrality measures}. This approach has proven to be very successful, with applications ranging from sorting the results of search engines~\cite{Page98,K99}, to improving the circulation of vehicles in modern cities~\cite{higham2017overview} and extracting and studying the evolution of communities in social networks~\cite{kanawati2014seed}. 
However, data typically have multiple features that may be overlooked by the standard graph representation. Thus, multilayer and temporal networks can be used instead to better capture such features. An explanatory  example is in the analysis of  scientific publications: a standard graph mining approach builds a citation network where nodes are authors citing each other. 
This approach then assigns importance to each node based on this mono-dimensional network structure. 
However, authors publish and cite in different journals which have different levels of importance. 
An alternative approach is thus to build a multilayer citation network where relations have the form: author $i$ publishes a paper in  journal $k$ and cites a paper in journal $\ell$ authored by $j$. 

\def\scr{\footnotesize}
\noindent\begin{minipage}{.45\columnwidth}
\begin{tikzpicture} 
		\node[thick,draw=black,circle,inner sep=4] (a) at (0,0) {};
		\node[above of = a,yshift=-1.5em,xshift=0em] {\scr{author $i$}};
		\node[thick,draw=black,circle,inner sep=4] (b) at (2,-.3) {};
		\node[above of = b,yshift=-1.5em,xshift=0em] {\scr{author $j$}};
		\draw[thick,-{>[scale=1.2]}] (a) to [out=320,in=200] node[yshift=.8em] {\scr{cites}} (b);
\end{tikzpicture}
\end{minipage}
\begin{minipage}{.45\columnwidth}
\begin{tikzpicture}
\node[thick,draw=black,circle,inner sep=4] (a) at (0,0) {};
\node[above of = a,yshift=-1.5em,xshift=-.7em] {\scr{author $i$}};
\node[circle,draw=black,thick,inner sep=17] (ja) at (-.3,0){};
\node[below of = ja,yshift=-.35em,xshift=0em] {\scr{journal $k$}};
\node[thick,draw=black,circle,inner sep=4] (b) at (2,-.3) {};
\node[above of = b,yshift=-1.5em,xshift=0em] {\scr{author $j$}};
\node[circle,draw=black,thick,inner sep=17] (jb) at (2.1,-.3){};
\node[below of = jb,yshift=-.35em,xshift=0em] {\scr{journal $\ell$}};
\draw[thick,-{>[scale=1.2]}] (a) to [out=320,in=200] node[yshift=.8em,xshift=-.1em] {\scr{cites}} (b);
\end{tikzpicture}
\end{minipage}

\noindent This allows for a better representation of the data and also provides the opportunity to measure the importance of both nodes (authors) and layers (journals). 
Moreover, a temporal aspect can be further added to keep track of {\it when} the citation took place. 

The new higher-order graph structure allows the possibility of computing importance scores that take into account multiple data features. 
  This construction, however,  introduces a number of challenges from both the mathematical and the computational point of view. Indeed, tensors are now needed to encode the network structure. Their use leads to a large increase in the problem size as well as to several numerical complications that arise when moving from linear to nonlinear operators. 

Several strategies have been proposed to translate centrality models for monolayer networks to higher-order settings. 
We review some of them in the next \S\ref{sec:related-work}. 
In this paper we focus on extending the well-established HITS algorithm for directed networks~\cite{K99}.  
From a mathematical viewpoint, 
HITS is based on the Perron eigenvector of the adjacency matrix of an undirected, bipartite network built from the directed network under study~\cite{BHM80,BEK13}. 
We extend this idea to the higher-order setting by 
defining a new eigenvector-based centrality measure, which we call \textit{Multi-Dimensional HITS} ({\it MD-HITS}), that assigns scores to nodes, layers and time stamps and is based on the Perron eigenvector of a multi-homogeneous order-preserving map \cite{gautier2017perron,gautier2017tensor}. 
We prove existence and uniqueness 
    of the centrality measure without requiring any connectivity assumption on the underlying  network. 
    This is of paramount importance and sets MD-HITS aside from other eigenvector-based centrality measures (see \S\ref{sec:related-work}), whose existence and uniqueness is  guaranteed only for strongly connected graphs, i.e., for networks whose adjacency matrices/tensors are irreducible, even though irreducibility does not hold for many real-world networks. MD-HITS, on the other hand, is {\it always} computable, regardless of the connectivity structure of the network.

The remainder of the paper is organized as follows. 
In \S\ref{sec:related-work} we review relevant related work. 
After reviewing the HITS algorithm for monolayer networks, in \S\ref{sec:model} we describe MD-HITS centrality and prove its existence, uniqueness, and maximality. 
In \S\ref{sec:algorithm} we describe the accompanying algorithm and prove its convergence.  
Numerical experiments on real-world networks are discussed in \S\ref{sec:experiments}.

\subsection{Related work}\label{sec:related-work}
Recent years have witnessed a growth in the number of 
centrality measures for the setting of multilayer networks. 
The framework that has drawn the most attention is probably that of multiplex networks: static multilayer networks where all the layers contain the same set of nodes and connections are allowed within layers but not across.  
For networks of this type, several centrality measures have been defined for the case of undirected layers. 
Eigenvector-type centrality measure were proposed in~\cite{BNL14,SRC13} following two complementary approaches: either by first computing the eigenvector centrality of the nodes in each layer and then aggregating the results, or by first aggregating the adjacency matrices of the different layers  and then computing the eigenvector centrality of the resulting network. 
The concept of eigenvector versatility was introduced in \cite{DSOGA15}. Here, the  multiplex network is embedded into a larger vector space and is represented by means of a supra-adjacency matrix that encodes information from all the layers.  
These approaches only define centrality vectors for nodes and do not address the problem of assigning scores to layers, which is tackled, e.g., in~\cite{tudisco2017node, sole2014centrality,rahmede2017centralities,zhou2007co,deng2009generalized,ng2011multirank,li2012har}.  
In \cite{rahmede2017centralities} the importance of both nodes and layers is defined via a parametric two-way recurrence built from the adjacency matrices of the individual layers. 
Multi-dimensional PageRank and HITS ranking methods for multiplex graphs based on third order adjacency tensors and their eigen and singular vectors have been proposed in \cite{zhou2007co,deng2009generalized,ng2011multirank,li2012har}, where  ``nodes" and ``layers" are sometimes referred to as ``objects" and ``relations". A related approach, designed for hypergraphs rather than multiplexes, is also discussed in \cite{benson2018three}.
In \cite{zhou2007co,deng2009generalized} the authors introduce Co-HITS: a model proposed to address the case of bipartite graphs where the content information and the relevance constraints come from both sides of the bi-partition. Following up on this approach, MultiRank and HAR algorithms were introduced in \cite{ng2011multirank,li2012har} as  multi-dimensional versions of PageRank and HITS respectively. Both these models define the importance of objects and relations in multi-relational data in terms of $Z$-eigen or singular vectors of normalized adjacency tensors. 
These centrality measures are well-defined only under restrictive assumptions on the connectivity of the network; specifically, they require all the individual layers to be strongly connected.  
Other improvements in the analysis of important components in multi-dimensional data are based on tensor factorization; see, e.g., \cite{sun2005cubesvd,kolda2005higher,rendle2009learning}. In \cite{li2016multivcrank} a tensor-based ranking scheme is applied to hypergraphs in order to develop a multi-visual-concept ranking scheme for image retrieval. 

To the best of our knowledge, none of these models is well-defined in the case of disconnected networks. Moreover, most of them cannot treat temporal networks or networks that contain edges across the layers.
In contrast, we compute the centrality of components in networks that may include inter-layer connections and that are allowed to change over time. 
The proposed approach builds on the concept of multi-homogeneous map \cite{gautier2017perron} and -- unlike other models -- is always well-defined and easily computable via a simple and globally convergent algorithm.

\section{The model}\label{sec:model}
A static monolayer network can be described as a set of nodes $V = \{1, \dots, n_V\}$ and a set of edges between them. 
Equivalently, it can be represented via its adjacency matrix $A=(A_{ij})\in\mathbb{R}^{n_V\times n_V}$, where $A_{ij}=\omega_{ij}>0$ is a  weight that quantifies the strength of edge $i\to j$, if present, and $A_{ij}=0$ otherwise.

Moving up in dimensionality, we can define a \textit{multilayer} network as a triplet: a set of nodes $V$, a set of layers $L = \{1, \dots, n_L\}$ on which these nodes ``live'', and a set of edges. 
Connections may exist between different nodes both within  and across layers. 
We further say that a network is a \textit{temporal} multilayer network if the edges are assigned a time label $t$, within a time window that we will assume discrete  $T = \{1, \dots, n_T\}$.
Thus, each edge in such a network is identified by two nodes, two layers and the time when the interaction takes place: 
node $i$ on layer $\ell$ connects to node $j$ on layer $k$ at time $t$.
Consequently, any temporal multilayer network can be represented via an \textit{adjacency tensor}~ $\A = (\A_{ij\ell kt})$~entry-wise~defined~as 
$$
\A_{ij\ell kt} = \begin{cases}
\omega_{ij\ell kt} &i \text{ on layer } \ell   \rightarrow j \text{ on layer } k, \text{ at time } t\\
0 & \text{otherwise}
\end{cases}
$$
for any $i,j \in V$, $\ell,k\in L$ and $t\in T$ and where, as before, $\omega_{ij\ell kt}>0$ quantifies the strength of the corresponding edge. 
A multilayer network is directed if there is at least one edge that is not reciprocated, i.e., 
if there exists at least a tuple of indices $(i,j,k,\ell,t)$ such that $\A_{ij\ell kt}\neq\A_{jik\ell t}$.
In the following, we refer to nodes, layers and time stamps as the \textit{components} of the network.

\subsection{HITS}\label{sec:monolayer-hits}
The standard HITS model for monolayer networks defines two types of importance for nodes: the {\it hub} and the {\it authority} scores. 
The former evaluates the importance of a node as a ``broadcaster'' whereas  the latter accounts for its relevance as a ``receiver'' of information. These two notions are related through mutually-reinforcing recursive relations: the importance $h_i\geq 0$ of node $i$ as a hub is proportional to the sum of the authority scores $a_j$ of all the nodes $j$ node $i$ points to. 
Vice-versa, the importance $a_i\geq 0$ of $i$ as an authority is proportional to the sum of all the hub scores $h_j$ of nodes $j$ that point to $i$. 
Using the adjacency matrix, we can describe these relations as 
\begin{equation}\label{eq:hits-monolayer}
    \lambda_1 h_i = \sum_j A_{ij}a_j \quad \text{ and } \quad \lambda_2 a_i = \sum_{j}A_{ji} h_j\, 
\end{equation}
for $i\in V$ and scalars $\lambda_1,\lambda_2>0$. 
If we let $\b h =(h_i)\in\RR^{n_V}$ and $\b a = (a_i)\in\RR^{n_V}$ 
then~\eqref{eq:hits-monolayer} rewrite as $\lambda_1 \b h = A\b a$ and $\lambda_2 \b a = A^T \b h$. 

\subsection{Multi-Dimensional HITS}
We propose an extension of the HITS model  to the framework of temporal directed multilayer networks. 
 Here relations occur at different time stamps and both within and across layers, thus making both nodes  and layers  play the roles of spreaders and gatherers of information. For this reason, 
 we consider two vectors of centrality for nodes: a vector $\b h \in \RR^{n_V}$ of hub scores and a vector $\b a\in\RR^{n_V}$ of authority scores, and two vectors that account for the broadcasting and receiving capability of layers: vectors $\b b \in\RR^{n_L}$ and $\b r\in\RR^{n_L}$, respectively. 
 Finally, time stamps inherit importance from the  relationships occurring at that time stamp and thus we define one vector $\b \tau \in \RR^{n_T}$ that encodes their importance. 
All these vectors are nonnegative and normalized so that their largest entry is 1; with this convention, each centrality score can be interpreted as a ``fraction of importance''.
We denote by $\bb c = (\b h, \b a, \b b, \b r, \b \tau) \in \RRR $ the tuple containing the five centrality vectors, where 
$
\RRR := \RR^{n_V}\times \RR^{n_V} \times \RR^{n_L}\times \RR^{n_L}\times \RR^{n_T}
$ and with the convention $\b c_1=\b h$, $\b c_2=\b a$, $\dots$, $\b c_5 =\b \tau$.

Hub and authority scores for nodes and layers are defined via mutually-reinforcing recursive relationships that involve the five vectors as follows. 
A node receives a high hub score if, at important time stamps, it originates several edges that leave it from important layers (in the sense of broadcast centrality) to reach authoritative nodes that lie on layers that have high receive centrality.
Similarly, a node receives a high authority score if, at important time stamps, it is the target -- on highly authoritative layers -- of many edges that originate from nodes that have a high hub score and lie on layers with high broadcast centrality. 
The broadcast and receive centrality indices for layers are formally defined in an analogous way. 
Finally, a time stamp is considered to be important if several links leave important nodes (in the sense of hub centrality) from layers with high broadcast centrality to target authoritative nodes on authoritative layers. 
These recursive relationships can be formalized by describing the entries of the tuple $\bb c$ in terms of the unique (normalized) Perron eigenvector of a suitable 
multi-homogeneous map \cite{gautier2017perron}, which we shall 
call $\F^{\b \alpha}_\A$, defined from the adjacency tensor of the multilayer network $\A$. 
Let us first consider the map  $\F_{\A}=(f_1,\dots, f_5):\RRR\to \RRR$  that acts on a tuple $\bb x = (\b x_1,\b x_2,\b x_3, \b x_4, \b x_5)\in\RRR$ as 
$$
\bb x \mapsto \F_{\A}(\, \bb x\,) = (f_1(\, \bb x\, ), \dots, f_5(\, \bb x\,)).
$$
The mappings $f_1,f_2:\RRR \to \RR^{n_V}$, $f_3,f_4:\RRR\to\RR^{n_L}$ and $f_5:\RRR\to\RR^{n_T}$ are particular tensor-vector products and define the ``slices'' of the multi-dimensional map $\F_\A$. Precisely, the $i_s$-th entry of  $f_s( \, \bb x \,)$ is defined by 
\begin{align*}
    \!\!\!
    \sum_{{\footnotesize \begin{array}{c}i_1,\dots,i_{s-1},\\ i_{s+1},\dots,i_5\end{array} }}\!\!\!\!\!\!
\A_{i_1\dots i_5}(\b x_1)_{_{i_1}}\!\cdots (\b x_{s-1})_{_{i_{s-1}}}\! (\b x_{s+1})_{_{i_{s+1}}}\!\cdots (\b x_{5})_{_{i_5}}
\end{align*}
for $s=1,\dots, 5$, where, in the above summations, $i_1,i_2\in\{1, \dots, n_V\}$, $i_3,i_4\in\{1, \dots, n_L\}$ and $i_5\in\{1, \dots, n_T\}$. 
Following~\cite{gautier2017perron}, we say that $\bb x\in \RRR$ is an eigenvector for $\F_\A$ with eigenvalue $\b \mu \in \RR^5$, if 
\begin{equation}\label{eq:multi-hits-noalpha}
\F_\A(\bb x) = \b \mu \otimes \bb x
\end{equation}
where $\b \mu = (\mu_1, \dots, \mu_5)$ and 
$\b \mu \otimes \bb x =(\mu_1 \b x_1, \dots, \mu_5 \b x_5)$.  

With this notation, the relationships that define the components of $\bb c$ can be rewritten in terms of a nonnegaitve eigenvector of $\F_{\A}$, namely $\F_\A(\bb c) = \b \lambda \otimes \bb c$
where $\b \lambda = (\lambda_1,\dots,\lambda_5)$ is a positive vector.  
However, in this way, the centrality $\bb c$  may not be well defined, as existence and uniqueness of a nonnegative eigenvector of $\F_\A$  cannot be  ensured for a general adjacency tensor $\A$. 
To avoid this critical drawback, we consider the following modification of  $\F_{\A}$: 
$$
\F_{\A}^{\b \alpha}(\,\bb x\,)=(f_1(\, \bb x \,)^{\alpha_1},\dots,   f_5(\,\bb x\,)^{\alpha_5})\, ,
$$
where $\b \alpha = (\alpha_1, \dots, \alpha_5)$ is such that $0<\alpha_i\leq 1$ for all $i=1, \dots, 5$, and the $\alpha_s$-th power of the vector $f_s(\, \bb x \,)$ is understood entry-wise 
$$f_s( \, \bb x \,)^{\alpha_s} = (f_s( \, \bb x \,)_{_{1}}^{\alpha_s}, \dots, f_s( \, \bb x \,)_{_{n_s}}^{\alpha_s}),\,\, s=1,\dots,5,$$
with $n_1 = n_2 = n_V$, 
$n_3 = n_4 = n_L$,  
and $n_5 = n_T$. 

We can now proceed with the definition of the multi-dimensional HITS centrality. 
\begin{definition}\label{def:multi-hits} 
Let $\A$ be the adjacency tensor of a temporal multilayer network  and let $\b \alpha = (\alpha_1,\dots, \alpha_5)$ be  such that $0<\alpha_s \leq 1$ for all $s=1, \dots, 5$. The   Multi-Dimensional HITS ({\it MD-HITS}) centrality  $\bb c = (\b h, \b a, \b b, \b r, \b \tau)\in \RRR$ is an entry-wise  nonnegative eigenvector of $F_\A^{\b \alpha}$, such that $\|\b c_1\|_\infty=\dots=\|\b c_5\|_\infty=~1$,i.e.\
\begin{equation}\label{eq:multi-hits}
\F_{\A}^{\b \alpha}(\bb c) = \b \lambda \otimes \bb c
\end{equation}
for some positive vector $\b \lambda=(\lambda_1,\dots, \lambda_5)$. 
\end{definition}
Note that \eqref{eq:multi-hits} generalizes \eqref{eq:multi-hits-noalpha}, since $F_\A^{(1,\ldots,1)} = F_\A$.  
Also note that we require the normalization condition $\|\b c_s\|_\infty=1$  in order to ensure the interpretation of centrality scores as  fraction of importance. We will see in the next section that, for a large range of parameters $\b \alpha$, MD-HITS centrality defined above exists, is unique and satisfies a maximality property analogous to that  of the Perron singular vectors of a nonnegative matrix. 
Moreover, we will describe the conditions under which the newly introduced centrality vectors have {\it positive} entries. 
\subsection{Existence, uniqueness and maximality of MD-HITS}
Let us start by pointing out that the vectors in $\bb c$ from Definition~\ref{def:multi-hits} may have zero entries. 
From~\eqref{eq:multi-hits} it is readily seen that this is the case when, for example, a node $i$ does not have outgoing edges from any layer and at any time stamp; indeed, in this case we  have  $\A_{ij\ell kt}=0$ for every $j \in V$, $\ell,k\in L$ and $t\in T$, and hence  $f_1(\, \bb x \,)_{_i} = 0$ for every $\bb x\in \RRR$. 
Our model then correctly assigns $h_i = 0$ to node $i$, since it is inactive as a hub in the multilayer. 
The same reasoning applies to inactive authority nodes, broadcast/receive layers, and time stamps, which will thus be appropriately assigned $a_j=b_\ell= r_k = \tau_t=0$. 
In any other situation, we want the centrality score of a component to be strictly positive, i.e., we want $\bb c= (\b h, \b a, \b b, \b r, \b \tau)\in\C_\A$, where 
$$
\mathcal C_{\A} \!= 
\left\{ \!\!\!\!\begin{array}{ll} 
\bb x \in \RRR : & \!\!\!\!\! \|\b x_s\|_\infty = 1,  \text{ for all }s=1,\dots,5, \text{ and }\\
                 & \!\!\!\!\!(\b x_s)_{_{i_s}} \!= 0  \, \text{ if } \sum_{\!\!\!\!{\scriptsize \begin{array}{c}i_1,\!...,i_{s-1},\\ i_{s+1},\!...,i_5\end{array} }}\A_{i_1,\dots,i_5}=0,  \\
                 & \!\!\!\!\!\text{or } (\b x_s)_{_{i_s}} > 0  \text{ otherwise.}
\end{array}\!\!\!\!\right\}
$$
Theorem \ref{thm:existence-uniqueness} below shows that, for any nonempty temporal multilayer network, MD-HITS centrality  
exists, is unique, and belongs to $\mathcal C_{\A}$ for appropriate $\b\alpha$. 
\begin{theorem}\label{thm:existence-uniqueness}
Let $\b \alpha = (\alpha_1,\ldots,\alpha_5)>0$ be such that $\rho(M_{\b \alpha})<1$, where $\rho(M_{\b \alpha})$ is the spectral radius of 
$$
\matrix = 
\begin{bmatrix}
0        & \alpha_2 & \alpha_3 & \alpha_4 & \alpha_5 \\
\alpha_1 & 0        & \alpha_3 & \alpha_4 & \alpha_5 \\
\alpha_1 & \alpha_2 & 0        & \alpha_4 & \alpha_5 \\
\alpha_1 & \alpha_2 & \alpha_3 & 0        & \alpha_5 \\
\alpha_1 & \alpha_2 & \alpha_3 & \alpha_4 & 0
\end{bmatrix}.
$$
If $\A$ is not the zero tensor, then there exist a unique $\bb c\in \mathcal C_{\A}$ and a unique $\b \lambda = (\lambda_1, \dots, \lambda_5)>0$ such that $\F_\A^{\b \alpha}(\bb c)=\b \lambda \otimes \bb c$. 
Moreover, there exists  
$\b \beta = (\beta_1, \dots, \beta_5)>0$ such that, 
$\lambda_1^{\beta_1}\cdots \lambda_5^{\beta_5} \geq |\mu_1^{\beta_1}\cdots\mu_5^{\beta_5}|$ for any other eigenvalue $\b \mu \in \RR^5$ of $\F_\A^{\b \alpha}$.  
\end{theorem}
\begin{proof}
The matrix $\matrix$ is irreducible and nonnegative, since $\alpha_s>0$ $\forall s$, and hence there exists a unique positive vector $\b \beta$ such that $\matrix \b \beta = \rho(M_{\b \alpha})\b \beta$ and $\sum_i \beta_i = 1$. 
Moreover, note that, by definition, any vector $\bb x\in \mathcal C_{\A}$ has a prescribed zero pattern forced by the adjacency tensor $\A$, i.e., there exist five sets of indices $O_1,O_2 \subseteq V$, $O_3, O_4\subseteq L$ and $O_5 \subseteq T$ such that 
$(\b x_s)_{_{i_s}} = 0$ if and only if $i_s\in O_s$, for $s=1, \dots, 5$. 
For any two $\bb x, \bb y\in \mathcal C_{\A}$, define the map
\begin{equation}\label{eq:hilbert-metric}
    d_H(\,\bb x, \, \bb y\,) := \sum_{s=1}^5 \beta_s \log \left( \max_{i_s\notin O_s}\frac{(\b x_s)_{_{i_s}}}{(\b y_s)_{_{i_s}}} \max_{i_s\notin O_s}\frac{(\b y_s)_{_{i_s}}}{(\b x_s)_{_{i_s}}}   \right)\, .
\end{equation}
This is a form of higher-order Hilbert metric such that the pair $(\mathcal C_{\A}, d_H)$ is a complete metric space (see, e.g., \cite[Prop.\ 2.5.4]{BookLN}). 
The map $d_H$ is a projective metric, i.e., it is invariant under scaling along any of the dimensions.   
Hence, $d_H(\, G(\,\bb x\,),\, G(\,\bb y\,)\, ) = d_H(\, F_{\A}(\,\bb x\,),\, F_{\A}(\,\bb y\,)\, )$ for any $\bb x, \bb y \in \mathcal C_{\A}$, where  
\begin{equation}\label{eq:G}
    \textstyle{G(\, \bb x\, ) = \left(\frac{f_1(\, \bb x\,)}{\|f_1(\,\bb x\,)\|_\infty}, \dots, \frac{f_5(\, \bb x\,)}{\| f_5(\,\bb x\,)\|_\infty}\right)}\, .
\end{equation}
Now note that for every $\b \mu = (\mu_1,\dots,\mu_5)>0$, every $\bb x\in \mathcal C_{\A}$, and every $s=1,\dots,5$ we have the following homogeneity equality 
$$f_s(\b \mu \otimes \bb x) = \frac{\mu_1^{\alpha_1}\cdots \mu_5^{\alpha_5}}{\mu_s^{\alpha_s}}f_s(\bb x)\, ,$$
where 
$\b \mu \otimes \bb x= (\mu_1\b x_1, \dots, \mu_5\b x_5)$. 
Moreover, for any $\bb x, \bb y\in\mathcal C_{\A}$, entry-wise we have $\b \gamma \otimes \bb y \leq \bb x \leq \b \delta \otimes \bb y$ where $\delta_s = \max_{i_s\notin O_s} (\b x_s)_{i_s}/(\b y_s)_{i_s}$ and $\gamma_s = \min_{i_s\notin O_s} (\b x_s)_{i_s}/(\b y_s)_{i_s}$. Therefore,   the following inequalities hold
$$
\frac{\gamma_1^{\alpha_1}\cdots \gamma_5^{\alpha_5}}{\gamma_s^{\alpha_s}}f_s(\bb y)\leq f_s(\bb x) \leq \frac{\delta_1^{\alpha_1}\cdots \delta_5^{\alpha_5}}{\delta_s^{\alpha_s}}f_s(\bb y)
$$
for any $s=1, \dots, 5$. We deduce that 
$$
d_H(\, F_{\A}(\,\bb x\,),\, F_{\A}(\,\bb y\,)\, ) \leq \left(\max_{s=1,\dots,5}\frac{(\matrix \b \beta)_s}{\beta_s}\right)  d_H(\,\bb x,\, \bb y\,)\, 
$$
for every $\bb x, \bb y \in \mathcal C_{\A}$. 
This, together with the Collatz--Wielandt formula for nonnegative matrices (see, f.i., \cite[Cor.\ 8.1.31]{BookHJ}) implies 
\begin{equation}\label{eq:CW}
    d_H(\, G(\,\bb x\,),\, G(\,\bb y\,)\, ) \leq \rho(M_{\b \alpha}) \,  d_H(\,\bb x,\, \bb y\,)\, .
\end{equation}
Finally, as $\rho(M_{\b \alpha})<1$, the above inequality implies that
the map $G$ is a strict Lipshitz contraction on the complete metric space $(\mathcal C_{\A}, d_H)$. 
By the Banach fixed point Theorem, there exists a unique ${\bb c}\in\mathcal C_{\A}$ such that $G(\, {\bb c}\, ) = {\bb c}$. 
By definition of $G$ this implies that there exists a unique $\b \lambda > 0$  such that $F_\A^{\b \alpha}(\bb c) = \b \lambda \otimes \bb c$ holds.   
The proof of the maximality of $\b \lambda$, i.e., of the fact that $\lambda_1^{\beta_1}\cdots \lambda_5^{\beta_5} \geq |\mu_1^{\beta_1}\cdots\mu_5^{\beta_5}|$ for any $\b \mu\in\RR^5$ such that $\F_{\A}(\, \bb x \,)=\b \mu \otimes \bb x$ holds  for some $\bb x \in \RRR$, follows directly from \cite[Thm.\ 4.1]{gautier2017perron} and is omitted here.
\end{proof}

Existence and uniqueness of MD-HITS centrality is thus ensured when $\b\alpha=(\alpha_1,\ldots, \alpha_5)$ is such that $\rho(M_{\b\alpha})<1$.
The following result provides a criterion for the selection of $\b \alpha$. 

\begin{theorem}\label{thm:hom-matrix}
Let $\b\alpha$ be such that $0<\alpha_s \leq 1$,~for~$s = 1,\ldots,5$.~If 
$(\alpha_1 + \cdots + \alpha_5) - \min_s\alpha_s \leq 1, $ 
then either $\rho(M_{\b\alpha}) < 1$ or $\rho(M_{\b\alpha})=1$ and $\alpha_s = 1/4$ for all $s$.
\end{theorem}
\begin{proof}
The result is a direct consequence of the Gershgorin Theorem  for irreducible matrices applied to $\matrix$; see, e.g., \cite{varga2010gervsgorin}. 
The theorem states that the eigenvalues of $\matrix$ lie within the union of the circles $\Delta_s = \{\lambda \in \mathbb C: |\lambda|\leq \sum_{i\neq s}\alpha_i\}$, $s=1, \dots, 5$, and, since $\matrix$ is irreducible, an eigenvalue of $\matrix$ cannot lie on the boundary of a disk $\Delta_s$ unless it lies on the boundary of every disk. 
\end{proof}

\subsection{Relation with tensor singular vectors}
A well known matrix-theoretic characterization of HITS centrality for monolayer networks is in terms of the dominant singular vectors of the adjacency matrix $A$. 
The following theorem shows that an analogous relation holds 
between the components of $\bb c$ and the singular vectors of the adjacency tensor $\A$ defined as in \cite{lim2005singular}.  
\begin{theorem}
Let $\b \alpha = (\alpha_1,\ldots,\alpha_5)>0$ be such that $\rho(M_{\b\alpha})<1$ and let $\bb c \in \mathcal C_{\A}$ be the corresponding MD-HITS centrality. Then there exists a positive real number $\sigma$ such that~${\lambda_1^{1/ \alpha_1}=\dots=\lambda_5^{1/ \alpha_5}=\sigma}$   and $\sigma$ is the maximal $\ell^{\big(\frac{\alpha_1+1}{\alpha_1},\cdots, \frac{\alpha_5+1}{\alpha_5}\big)}$-singular value of $\A$, with corresponding singular vectors $\b c_1, \dots, \b c_5$.
\end{theorem}
\begin{proof}
It follows by combining  Theorem~\ref{thm:existence-uniqueness} with \cite[Lemma~5.1]{gautier2017tensor}.
\end{proof}

\section{The algorithm}\label{sec:algorithm}     %
We present in Alg.\ref{alg:1} an efficient and parallelizable iterative method for the computation of the MD-HITS centrality tuple $\bb c$ defined in Definition \ref{def:multi-hits}. 
In the remainder of the paper,  we write $\|\bb x\|_{\b\beta}$, for any given $\bb x\in\RRR$, to denote the norm  
\begin{equation}\label{eq:norm}
\|\bb x\|_{\b\beta} = \sum_{s=1}^5 \beta_s \|\b x_s\|_\infty = \sum_{s=1}^5\beta_s \max_{i_s}|(\b x_s)_{i_s}|,
\end{equation}
where $\b\beta = (\beta_1,\ldots,\beta_5)>0$.

Note that each of the steps 2--6 in  Alg.~\ref{alg:1}, as well as each individual normalization at step 7, can be performed in parallel at each iteration, significantly enhancing the performance of the algorithm. 

The following theorem shows global convergence of the algorithm and provides an estimate of the number of iterations required to achieve convergence. %
\begin{theorem}\label{thm:convergence}
Let $\b \beta = (\beta_1,\ldots,\beta_5)>0$ be such that $\matrix \b \beta = \rho(M_{\b\alpha})\matrix$, with $\|\b\beta\|_1 = 1$. 
For $\bb c^{(0)}>0$, let $\bb c^{(k)} = (\b c_1^{(k)}, \dots, \b c_5^{(k)}) \in \RRR$ be defined as in Alg.~1. Then $\bb c^{(k)}\in \mathcal C_\A$ and $\lim_{k\to\infty}\bb c^{(k)}=\bb c$,  the MD-HITS centrality tuple. 
Moreover, for $k=0,1,2,\dots$, it holds
$$ 
\frac{\|\bb c^{(k+1)} - \bb c^{(k)}\|_{\b\beta}}{\|\bb c^{(k+1)}\|_{\b\beta}} \leq 2\,  \rho(M_{\b\alpha})^k\,\| \log (\bb c^{(1)}/\bb c^{(0)})\|_{\b \beta}
$$
where both the logarithm and the division in the right hand side are intended entry-wise, with the convention that $\log(0)=0$. 
\end{theorem}
\begin{proof}
Let $d_H$ and $G$ be defined as in  \eqref{eq:hilbert-metric} and \eqref{eq:G}, respectively. Then step 7 in Algorithm 1 rewrites as $\bb c^{(k+1)} = G(\bb c^{(k)})$. From this it follows that $\bb c^{(k+1)}\in \mathcal C_{\A}$ for any $k$ and also that, using \eqref{eq:CW}, 
\begin{gather}
    d_H(G(\bb c^{(k)}),\bb c^{(k)}) = d_H(G(\bb c^{(k)}), G(\bb c^{(k-1)}))\notag\\ \leq \rho(M_{\b \alpha}) d_H(\bb c^{(k)}, \bb c^{(k-1)})\leq \rho(M_{\b \alpha})^k d_H(\bb c^{(1)}, \bb c^{(0)}).\label{eq:aaa}
\end{gather}
This implies that $\bb c^{(k)}$ converges to the fixed point $\bb c$ of $G$, which is then the positive eigenvector $F_\A^{\b \alpha}(\bb c) = \b \lambda \otimes \bb c$. 
Now, given a vector $\b x$ let us denote by $\tilde{\b x}$ the vector with entries $(\tilde{\b x})_i = \log((\b x)_i)$ if $(\b x)_i>0$ and $(\tilde{\b x})_i=0$ otherwise. Note that  $|a-b|\leq \max\{a,b\}|\log a - \log b|$ for any two scalars $a,b>0$. Thus, for any $s=1,\dots,5$ and  any $\bb x,\bb y\in \mathcal C_{\A}$  we have 
\begin{align*}
\|\b x_s - \b y_s\|_\infty \leq \frac{\|\b x_s-\b y_s\|_{\infty} }{\displaystyle{\max_{i_s\notin O_s} \max\{(\b x_s)_{_{i_s}}, (\b y_s)_{_{i_s}}\} }}  \leq \|\tilde{\b x}_s -\tilde{\b y}_s\|_{\infty},
\end{align*}
since the entries of any vector in $\mathcal C_{\A}$ are at most $1$. Therefore,
\begin{align*}
    \|\b x_s - \b y_s\|_\infty &\leq \|\tilde{\b x}_s -\tilde{\b y}_s\|_{\infty} 
     = \log\left( \max_{i_s\notin O_s} e^{|(\tilde{\b x}_s)_{_{i_s}} - (\tilde{\b y}_s)_{_{i_s}} |}   \right) \\
    &= \log \left(\max\{\max_{i_s\notin O_s}\frac{(\b x_s)_{_{i_s}}}{(\b y_s)_{_{i_s}}}, \max_{i_s\notin O_s}\frac{(\b y_s)_{_{i_s}}}{(\b x_s)_{_{i_s}}} \} \right) \\
    &\leq \log \left( \max_{i_s\notin O_s}\frac{(\b x_s)_{_{i_s}}}{(\b y_s)_{_{i_s}}} \max_{i_s\notin O_s}\frac{(\b y_s)_{_{i_s}}}{(\b x_s)_{_{i_s}}}   \right), 
\end{align*}
and from \eqref{eq:hilbert-metric} it follows $\sum_s \beta_s \|\b x_s - \b y_s\|_\infty \leq d_H(\bb x,\bb y)$. 
This, together with \eqref{eq:aaa} and the fact that  $\|\b c^{(k+1)}_{s}\|_\infty =1$ for all $s=1, \dots,5$, shows that 
$$
\frac{\|\bb c^{(k+1)} - \bb c^{(k)}\|_{\b\beta}}{\|\b c^{(k+1)}\|_{\b\beta}} \leq   \rho(M_{\b \alpha})^k d_H(\bb c^{(1)}, \bb c^{(0)})\, .
$$
To conclude, we now show that 
\begin{equation}\label{eq:xyz}
    d_H(\bb x, \bb y) \leq 2 \|\log(\bb x/\bb y)\|_{\b \beta}\, .
\end{equation}
For any $s=1,\dots,5$ we have 
\begin{align*}
    &d_H(\b x_s,\b y_s) =\log \left( \max_{i_s\notin O_s}\frac{(\b x_s)_{_{i_s}}}{(\b y_s)_{_{i_s}}} \max_{i_s\notin O_s}\frac{(\b y_s)_{_{i_s}}}{(\b x_s)_{_{i_s}}}   \right) \\
    &= \max_{i_s\notin O_s}(\log x_{i_s}-\log y_{i_s} )+\max_{i_s\notin O_s}(\log y_{i_s}-\log x_{i_s}) \\
    & \leq 2 \max\{\max_{i_s\notin O_s}(\log x_{i_s}-\log y_{i_s} ), \max_{i_s\notin O_s}(\log y_{i_s}-\log x_{i_s})\}\\
    &= 2\|\log(\b x_s/\b y_s)\|_\infty
\end{align*}
which proves the desired bound \eqref{eq:xyz}.
\end{proof}

\SetKwBlock{Repeat}{For $k=0,1,2,3,\dots$ repeat}{{until $\|\bb c^{(k)}-\bb c^{(k+1)}\|_{\b\beta}/\|\bb c^{(k+1)}\|_{\b\beta}< \varepsilon$}}
\begin{algorithm2e}[t]
		\caption{MD-HITS algorithm}\label{alg:1}
		 \DontPrintSemicolon
		\KwIn{$\A$; $\b\alpha$, $\b\beta>0$ such that $\matrix\b \beta = \rho(M_{\b\alpha})\b \beta$ with  $\rho(M_{\b\alpha})<1$ and $\sum_i \beta_i =1$; tolerance $\varepsilon>0$; $\F^{\b\alpha}_{\A}=(f_1^{\alpha_1},\dots,f_5^{\alpha_5})$; initial guess $\bb c^{(0)} >0$. }
		\Repeat{
		$\b h      = f_1(\, \bb c^{(k)})^{\alpha_1}$ \;
		$\b a      = f_2(\, \bb c^{(k)})^{\alpha_2}$ \;
		$\b b      = f_3(\, \bb c^{(k)})^{\alpha_3}$ \;
		$\b r      = f_4(\, \bb c^{(k)})^{\alpha_4}$ \;
		$\b \tau   = f_5(\, \bb c^{(k)})^{\alpha_5}$ \;
		$\bb c^{(k+1)} = \left(\frac{\b h}{\|\b h\|_\infty}, \frac{\b a}{
		\|\b a\|_\infty},\frac{\b b}{\|\b b\|_\infty},\frac{\b r}{\|\b r\|_\infty},\frac{\b \tau}{\|\b \tau\|_\infty}\right)$\;
		}
		\KwOut{Approximation $\bb c^{(k+1)}$ to $\bb c$.}
\end{algorithm2e}

Thm.\ \ref{thm:convergence} shows that the number of 
 iterations $k_*$ required by Alg.\ \ref{alg:1} to achieve convergence decays as the inverse of the logarithm of $\rho(M_{\b \alpha})$, i.e., $k_* \approx s_1/\ln\rho(M_{\b \alpha})+s_2$, for some scalars $s_1,s_2$. This behaviour is confirmed by Fig.\ \ref{fig:iter-time-alpha}. 
This has computational relevance: In the case of sparse networks, each iteration of Alg.\ \ref{alg:1} requires $O(n_Vn_Ln_T)$ flops. However, Thm.\ \ref{thm:convergence} shows that tuning $\b \alpha$  allows us to reduce the overall timing and number of iterations required to compute the MD-HITS centrality vector. The dependence of the centrality obtained with respect to the variation of $\b \alpha$ is analyzed in \S \ref{sec:citation_network}.

\section{Experiments}\label{sec:experiments}
In this section we describe a number of numerical experiments on small networks and real-world data to provide insights in the performance of our model and algorithm. 
All the experiments were performed using MATLAB Version 9.1.0.441655 (R2016b)
on an HP EliteDesk running Scientific Linux 7.3 (Nitrogen), a 3.2 GHz Intel Core i7 processor, and 4 GB of RAM. We used a serial implementation of Alg.~\ref{alg:1}. 
Both the code and the data used in this paper can be found at: {\small\tt https://github.com/ftudisco/multi-dimensional-hits}.

\subsection{The curse of disconnectedness}
We show here that MD-HITS is able to correctly identify hubs and authorities  when applied to monolayer networks that do not satisfy the hypothesis required by HITS.

The eigenvector-based centrality measures for mono and multilayer networks considered in \S\ref{sec:related-work} all suffer ``the curse of disconnectedness'': uniqueness of the centrality scores is not guaranteed unless the network under study is strongly connected. 
Only local convergence of standard algorithms can be ensured for these models and different runs of the same method can lead to different results  \cite{farahat2006authority}.
As an example, consider the network in Fig.\ref{fig:example_curse_connectivity}.
%
\usetikzlibrary{positioning}
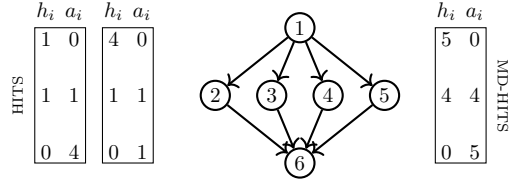
\begin{figure}[!h]
\centering
\vspace{-10pt}\hspace{-10pt}
\begin{tikzpicture}[scale=.75]
\tikzset{every node/.style={scale=.8,thick,inner sep=2}}
    \node[circle,draw=black] (1) at (4,1.2) {1};
    \node[circle,draw=black] (2) at (2.5,0) {2};
    \node[circle,draw=black] (3) at (3.5,0) {3};
    \node[circle,draw=black] (4) at (4.5,0) {4};
    \node[circle,draw=black] (5) at (5.5,0) {5};
    \node[circle,draw=black] (6) at (4,-1.2) {6};
    \path[->,thick] (1) edge (2) (1) edge (3) (1) edge (4) (1) edge (5) (2) edge(6) (3) edge (6) (4) edge (6) (5) edge (6);

    \node (hh) at (-.5,1.5) {$h_i$};
    \node (h1) at (-.5,1) {1};
    \node (h2) at (-.5, 0) {1};
    \node (h3) at (-.5,-1) {0};
    \node (aa) at (0,1.45) {$a_i$};
    \node (a1) at (0,1) {0};
    \node (a2) at (0, 0) {1};
    \node (a3) at (0,-1) {4};
    \draw (-.7,-1.2) rectangle (.2, 1.2);
    \node [rotate = 90] at (-1,0) {{\scriptsize HITS}};

    \node (hh) at (.7,1.5) {$h_i$};
    \node (h1) at (.7,1) {4};
    \node (h2) at (.7, 0) {1};
    \node (h3) at (.7,-1) {0};
    \node (aa) at (1.2,1.45) {$a_i$};
    \node (a1) at (1.2,1) {0};
    \node (a2) at (1.2, 0) {1};
    \node (a3) at (1.2,-1) {1};
    \draw (.5,-1.2) rectangle (1.4, 1.2);

    \node (hh) at (6.6,1.5) {$h_i$};
    \node (h1) at (6.6,1) {5};
    \node (h2) at (6.6, 0) {4};
    \node (h3) at (6.6,-1) {0};
    \node (aa) at (7.1,1.45) {$a_i$};
    \node (a1) at (7.1,1) {0};
    \node (a2) at (7.1, 0) {4};
    \node (a3) at (7.1,-1) {5};
    \draw (6.4,-1.2) rectangle (7.3, 1.2);
    \node [rotate = -90] at (7.6,0) {{\scriptsize MD-HITS}};
\end{tikzpicture}
\vspace{-12pt}\hspace{-10pt}
\caption{Example of the course of disconnectedness}
\label{fig:example_curse_connectivity}
\end{figure}
Depending on the initialization, the HITS algorithm for this graph returns two different solutions, displayed in the tables on the left-hand side of Fig.~\ref{fig:example_curse_connectivity} (nodes $2-5$ are isomorphic and thus assigned the same scores). 
In the first case the hub vector fails to detect that node $1$ is a better hub than nodes $2-5$. Similarly, in the second case, the authority vector fails to identify node $6$ as a the best authority. 
In practice, to overcome this potential ambiguity,  HITS requires a preprocessing phase where  the connectivity pattern of the data is investigated and, possibly, irreducibility is enforced by subsampling the data or perturbing the graph by adding artificial edges. 
This process can be extremely computationally demanding, especially in the setting of temporal multilayer networks.

MD-HITS, on the other hand, is always uniquely defined, even for monolayer graphs. In particular, when tailored to the monolayer setting, our model reduces to a ``nonlinear'' version of the classical HITS algorithm, which however does not suffer the curse of disconnectedness.  
Precisely, let $A$ be the adjacency matrix of a monolayer network that we understand as a temporal multilayer with $n_L=n_T=1$. Then, the MD-HITS centrality $\F_A^{\b \alpha}(\bb c) = \b \lambda \otimes \bb c$ is the unique positive solution in $\mathcal C_A$ of the system of equations
\begin{equation}\label{eq:nonlinearHITS}
    (A\b a)^{\alpha_1} = \lambda_1 \b h, \qquad (A^T \b h)^{\alpha_2} = \lambda_2 \b a \, .
\end{equation}
When $\alpha_1=\alpha_2=1$, the standard hub and authority score of the HITS model are retrieved and their uniqueness is not ensured. Instead, with the same proof of Theorem \ref{thm:existence-uniqueness}, if $M_{\b \alpha}=\begin{bsmallmatrix}0 & \alpha_2 \\ \alpha_1 & 0\end{bsmallmatrix}$ with $\rho(M_{\b\alpha})<1$, then \eqref{eq:nonlinearHITS} has a unique solution which we compute with Alg.\ref{alg:1}. In particular, if we apply Alg.\ref{alg:1} to the graph in Fig.\ref{fig:example_curse_connectivity} with $\alpha_1=\alpha_2=1/3$ we obtain the centrality scores displayed in the rightmost table of Fig.~\ref{fig:example_curse_connectivity} 
regardless of the starting point. 
These two vectors capture the actual roles of nodes in this graph. 

\subsection{HITS vs MD-HITS}
As we have seen in the 
example discussed in the introduction, taking into account multiple data features allows us to build a multilayer network out of a given dataset, rather than just a standard monolayer graph. We now show that when connectivity is ensured, MD-HITS returns the same rankings as HITS in the monolayer setting. Moreover, we show that allowing the modelling network to account for more 
facets of the data we are able to better detect the roles of nodes. To this end we consider the small example dataset presented in Fig.~\ref{fig:hitsvsmulti}. The graph on the left shows a multilayer network of interactions between four nodes. On the right, we display the standard monolayer network between the same four nodes that corresponds to the aggregate network associated to the multilayer. Here, the edge set is obtained by ignoring the layer aspect in the edges of the multilayer.  
It is easy to see that the hub and authority centrality retrieved by standard HITS on the monolayer network will assign the same scores to nodes $1$ and $2$ and to nodes $3$ and $4$, with the first pair ranked higher than the second. This result is confirmed by the first table of  Fig.~\ref{fig:hitsvsmulti}. 
The same ranking is obtained if MD-HITS (in the formulation of~\eqref{eq:nonlinearHITS}) is applied to the aggregate network,
regardless of the choice of exponents; the first table of Fig.~\ref{fig:hitsvsmulti} reports the scores obtained for $\alpha_1=\alpha_2=1/3$.  

\begin{figure}[t]
\begin{center} 
\begin{tikzpicture}[scale=.9]
    \draw[gray,thick,rounded corners=10pt] (-.3,-.3) rectangle (1.3,1.3);
    \node at (.6,-.45) {{\scriptsize Layer 1}};
    \draw[gray,thick,rounded corners=10pt] (1.5,-.3) rectangle (3.1,1.3);
    \node at (2.4,-.45) {{\scriptsize Layer 2}};
    \draw[gray,thick,rounded corners=10pt] (3.3,-.3) rectangle (4.9,1.3);
    \node at (4.1,-.45) {{\scriptsize Layer 3}};
    \draw[gray,thick,rounded corners=10pt] (-.6,-.6) rectangle (5.2,1.5);
    \node at (1.1,1.7) {{\scriptsize \color{red}{Multilayer network}}};
    \node[scale=.6,circle,draw=black,thick] (21) at (0,0) {2};
    \node[scale=.6,circle,draw=black,thick] (11) at (0,1) {1};
    \node[scale=.6,circle,draw=black,thick] (41) at (1,0) {4};
    \node[scale=.6,circle,draw=black,thick] (31) at (1,1) {3};
    \node[scale=.6,circle,draw=black,thick] (22) at (1.8,0) {2};
    \node[scale=.6,circle,draw=black,thick] (12) at (1.8,1) {1};
    \node[scale=.6,circle,draw=black,thick] (42) at (2.8,0) {4};
    \node[scale=.6,circle,draw=black,thick] (32) at (2.8,1) {3};
    \node[scale=.6,circle,draw=black,thick] (23) at (3.6,0) {2};
    \node[scale=.6,circle,draw=black,thick] (13) at (3.6,1) {1};
    \node[scale=.6,circle,draw=black,thick] (43) at (4.6,0) {4};
    \node[scale=.6,circle,draw=black,thick] (33) at (4.6,1) {3};
    \path[->, thick] (21)edge[] node[]{} (11) (41)edge[] node[]{} (11) (12)edge[] node[]{}(21)
    (12)edge[] node[]{} (41) (12)edge[] node[]{} (32) (22)edge[] node[]{} (42) (32)edge[] node[]{} (13)
    (32)edge[] node[]{} (23) (23)edge[] node[]{} (33) (43)edge[] node[]{} (23); 
    \draw[gray,thick,rounded corners=10pt] (5.7,-.3) rectangle (7.3,1.3);
    \node at (6.5,1.7) {{\scriptsize \color{red}{Aggregate}}};
        \node[scale=.6,circle,draw=black,thick] (2) at (6,0) {2};
    \node[scale=.6,circle,draw=black,thick] (1) at (6,1) {1};
    \node[scale=.6,circle,draw=black,thick] (3) at (7,0) {4};
    \node[scale=.6,circle,draw=black,thick] (4) at (7,1) {3};
    \path[<->, thick] (1)edge[] node[]{} (2) (2)edge[] node[]{} (3) (3)edge[] node[]{}(1)
    (1)edge[] node[]{} (4) (4)edge[] node[]{} (2); 
\end{tikzpicture}

\begin{tabular}{|r|llll|}
\hline
     HITS   &  1 & 1 & 0.86 & 0.86 \\
     MD-HITS & 1 & 1 & 0.78 & 0.78  \\
     \hline
\end{tabular}

\vspace{.5em}

\begin{tabular}{|l|llll|l|lll|}
\hline
    $\b h$ & 1 & 0.97 & 0.94 & 0.90 & $\b b$ & 0.81 & 1 & 0.80  \\
    $\b a$ & 0.98 & 1 & 0.91 & 0.93 & $\b r$ & 1 & 0.88 & 0.98 \\
    \hline
\end{tabular}
\end{center}
\vspace{-1em}
\caption{Top: Example multilayer network and its aggregate version. Center: HITS and MD-HITS node scores for the aggregate network. Bottom: MD-HITS node and layer scores.}
\vspace{-.5em}
\label{fig:hitsvsmulti}
\end{figure}
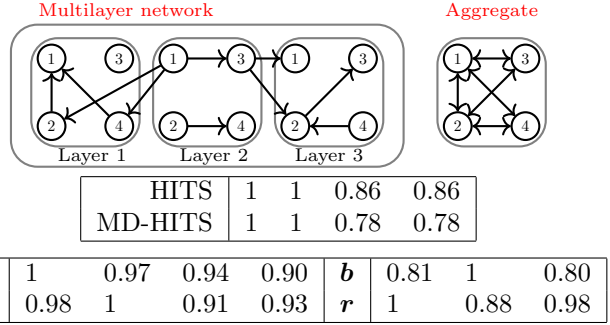

On the other hand, if we consider the multilayer network  of interactions and compute the five ranking vectors of MD-HITS for it, we obtain the results displayed in the bottom table in Fig.~\ref{fig:hitsvsmulti}. These reveal that, for example, node $1$ is a better hub than node $2$, which is expected since it is the node with the largest number of outgoing links originating from layer $2$, which is the most influential as a broadcaster. 

\subsection{Synthetic random data}
We investigate here the scalability of Algorithm 1, with respect to the size of the data. To this end, we use the tensor toolbox (v.\ 2.6) from \cite{TTB_Software} to build sparse random networks of increasing size, with $n_V = n_L = 25, 100, 200, \dots, 500, 1000, 2000, \dots, 5000$, $n_T = n_V^{1/3}$ and $n_Vn_L$ nonzeros (which correspond to a density of $n_V^{-3}$). For each of these networks, we computed the multi-dimensional HITS centrality tuple $\bb c$ and reported execution time (in seconds) and number of iterations. 
In the stopping criterion of Algorithm 1, we used the norm $\|\cdot\|_{\b\beta}$ defined for any given $\bb x\in\RRR$ as 
\begin{equation*}
\|\bb x\|_{\b\beta} = \sum_{s=1}^5 \beta_s \|\b x_s\|_\infty = \sum_{s=1}^5\beta_s \max_{i_s}|(\b x_s)_{i_s}|,
\end{equation*} where $\b \beta>0$ such that $\sum_s \beta_s =1$ is the eigenvector of $\matrix$ associated to $\rho(M_{\b \alpha})$: $\matrix \b \beta = \rho(M_{\b \alpha})\b \beta$ for a  uniform choice of $\b\alpha = \overline{\b\alpha} =  (1,1,1,1,1)/5$.
We set the tolerance to $\varepsilon= 10^{-6}$ and select as starting vector the vector of all ones: $\bb c^{(0)} = {\b 1}$. 
We iterated this process $100$ times and averaged the results. 
Fig.~\ref{fig:performances} reports the average timings required for the computation with the errorbar representing the standard deviation from the mean, versus $n_V$. 
The numbers represent the average number of iterations required to achieve convergence, rounded to the nearest integer. 
This figure clearly shows the outstanding performance of our {\it serial} implementation of the algorithm, showcasing its applicability to much larger datasets. 

\begin{figure}[!t]
    \centering
    \includegraphics[width=.95\columnwidth,clip,trim=0 1em 0 0]{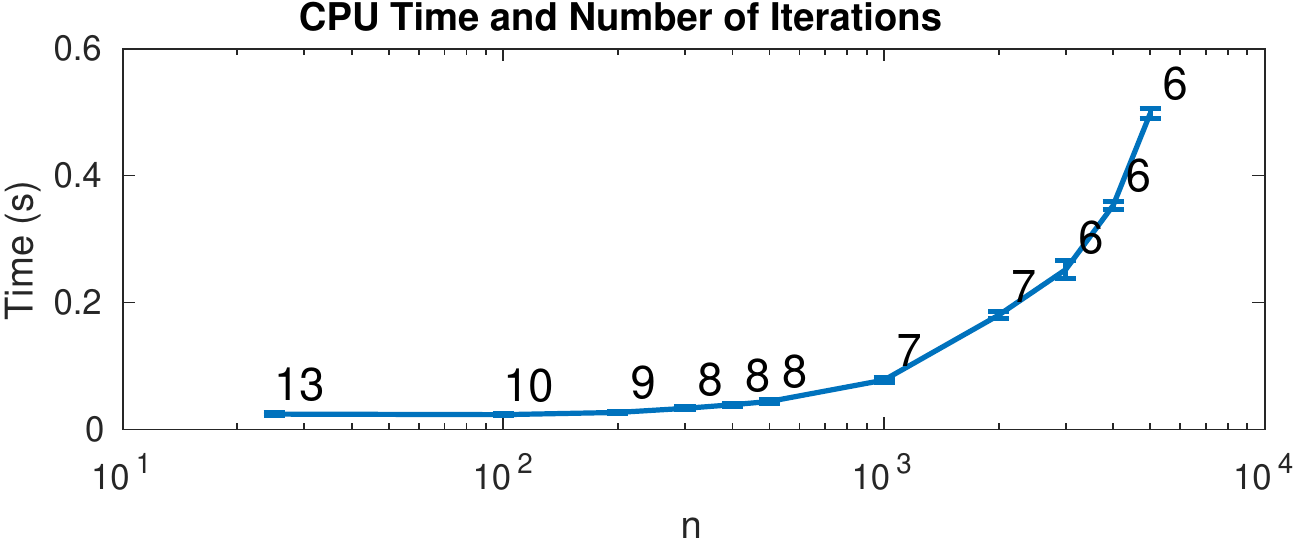}
    \caption{Average time (with standard deviation) required by Algorithm 1 to achieve convergence ($\varepsilon= 10^{-6}$) for a set of sparse random tensors of increasing size (horizontal axis). Numbers  show the average number of iterations.}
    \label{fig:performances}
\end{figure}

\subsection{Temporal multilayer citation network}\label{sec:citation_network}
In this section we perform experiments on a large real-world dataset of scientific publications. 
We build a multilayer temporal citation network from the scientific publications dataset available at~\cite{Tang:08KDD,citnetdata} (release: 2010-05-15) as follows: an edge goes from node $i$ on layer $\ell$ to node $j$ on  layer  $k$ at a given time  $t$ if, in year $t$, node $i$ authored a paper in journal $\ell$ and in that paper $i$ cites a paper authored by $j$~in~journal~$k$. 
The resulting adjacency tensor has $n_V = 592,373$ nodes (i.e., authors), $n_L = 12,608$ layers (i.e., journals), $n_T = 65$ time stamps (i.e., years) and $3,587,948$ nonzeros (i.e., citations). 
For this dataset we analyze numerically the stability of the ranking model with respect to changes in the choice of the exponents $\b \alpha$. 

Our experiments will show that  when no empirical knowledge is available to suggest otherwise, a reasonable choice for $\b\alpha$ is given by a uniform vector $\alpha \b 1$, $\b 1 = (1,\dots,1)$.
Since $\rho(M_{\alpha \b 1})=4\,\alpha$, in what follows we set  ${\b\alpha}^{(0)} = {\b 1}/5$ in order to ensure $\rho(M_{{\b \alpha}^{(0)}})=4/5 <1$.

\subsection*{Execution time and iteration count}
With this first set of experiments we want to show feasibility of Alg.\ref{alg:1}. We 
randomly selected  ten vectors $\b\alpha^{(i)}$, $i=1,\ldots,10$ such that $\rho(M_{\b\alpha^{(i)}})<1$ for all $i$ labelled so that   
$
\|{\b \alpha}^{(0)} -\b \alpha^{(1)}\|_2\geq \cdots \geq \|{\b \alpha}^{(0)} -\b \alpha^{(10)}\|_2 \, .
$
We computed the MD-HITS centrality via Alg.\ref{alg:1} for all the eleven choices of the exponents. 
Number of iterations and execution time for each of these vectors are shown in the table on the left of Fig.~\ref{fig:iter-time-alpha}. The right-hand plot of Fig.~\ref{fig:iter-time-alpha} shows how the number of iterations varies when $\b \alpha=\b 1\alpha$ and $\alpha$ ranges in $\{0.04,0.06, \dots,0.2\}$ (red dots). The black line plots the curve $c_1\log(\rho(M_{\b\alpha}))+c_2$ and is used to confirm the behavior predicted by Theorem \ref{thm:convergence}.    
\begin{figure}
    $\,$\hfill
    \begin{minipage}{.4\columnwidth}
    \vspace{-.5em}
     {\tiny \begin{tabular}{|c|cc|}
     \hline
        $\b \alpha^{(i)}$ & time(s) & it \\
        \hline
        $0$  & 180.1 & 18  \\
        $1$  & 212.6 & 18  \\
        $2$  & 236.1 & 20  \\
        $3$  & 187.8 & 16  \\
        $4$  & 199.7 & 17  \\
        $5$  & 235.5 & 20  \\
        $6$  & 352.9 & 30  \\
        $7$  & 282.4 & 24  \\
        $8$  & 129.3 & 11  \\
        $9$  & 293.8 & 25  \\
        $10$  & 223.4 & 19 \\
        \hline
        mean & 230.3 & 20\\
        \hline
     \end{tabular} }\end{minipage}
     \hfill
    \begin{minipage}{.4\columnwidth}\includegraphics[width=1\columnwidth]{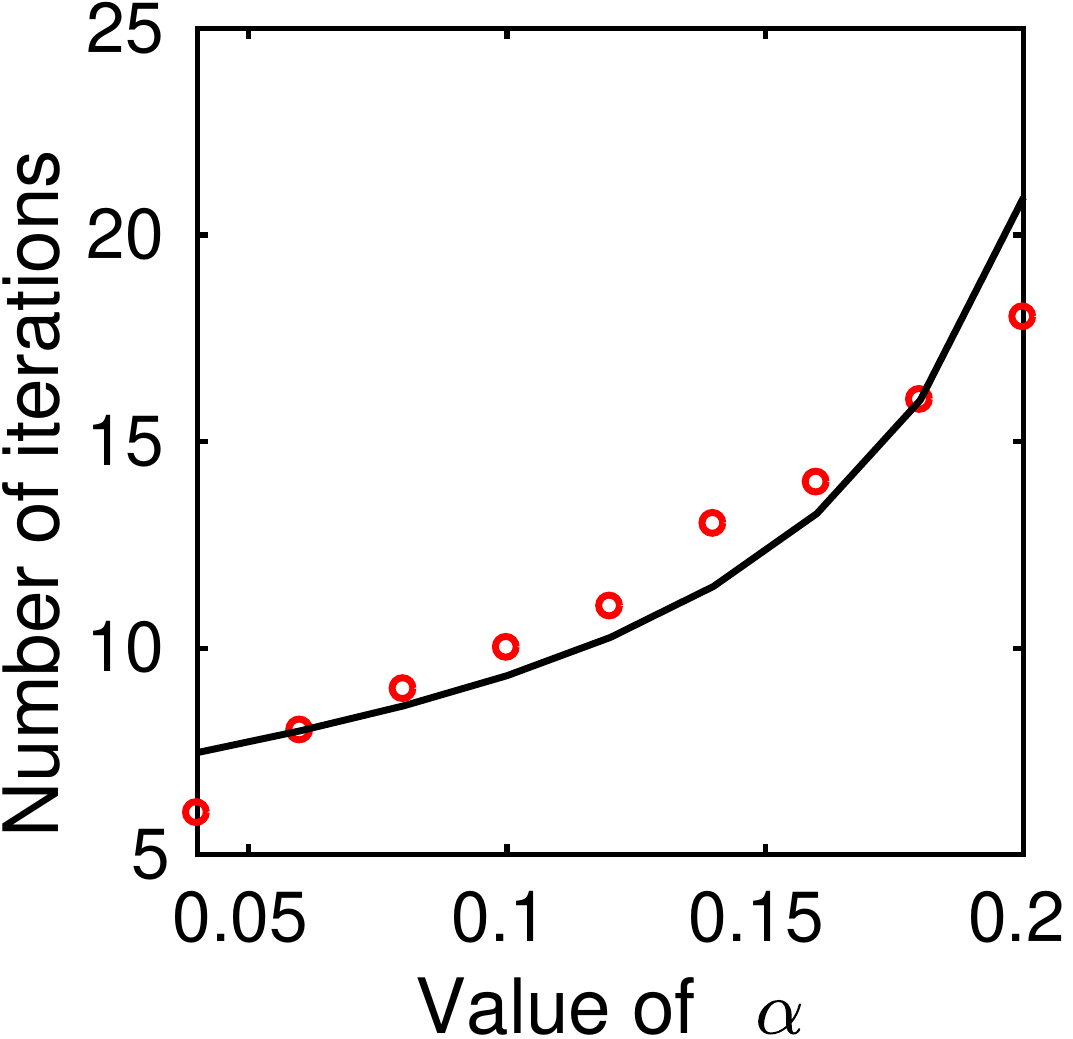}\end{minipage}\hfill$\,$
    \caption{Multilayer citation network. Left: execution times and number of iterations required for convergence of Alg.\ref{alg:1} for different $\b\alpha$s. Right: number of iterations required for convergence of Alg.\ref{alg:1} for ${\b \alpha}^{(0)} = \alpha\b 1$ for $\alpha\in\{0.04,0.06, \dots,0.2\}$ (circle) and predicted behavior (line); cf. Thm~\ref{thm:convergence}.}
    \label{fig:iter-time-alpha}
\end{figure}
Clearly, our algorithm requires only a small number of iterations and just a few seconds to compute the five centrality vectors for this dataset. 
\subsubsection*{Worst-case scenario analysis}
We now proceed to compare the derived rankings for the different choices of $\b\alpha$. 
 We restrict the following analysis to the worst-case setting by only considering the five exponents with the farthest distance from $\b\alpha^{(0)}$, i.e., $\b\alpha^{(1)},\ldots,\b\alpha^{(5)}$.
Experiments on the whole set $\b\alpha^{(1)},\ldots,\b\alpha^{(10)}$, not displayed here, showed  notably better results.
To compare rankings, 
we use the measure $\II_K:= 1-{\rm isim}_K(\L^{1},\L^2)$, where ${\rm isim}_K(\L^{1},\L^2)$ is the top $K$ intersection similarity   between the ranking vectors $\L^{1}$ and $\L^2$. 

The intersection similarity is a measure used to compare the top $K$ entries of two ranked lists that 
may not contain the same elements.  It is defined as follows: let $\L^{1}$ and $\L^2$ be two 
ranked lists, and let us call $\L^{j}_i$ the list of the top $i$ elements listed in $\L^{j}$, for  $j=1,2$. 
Then, the {\it top $K$ intersection similarity between $\L^{1}$ and $\L^2$} is defined as 
\begin{equation*}
{\rm isim}_K(\L^{1},\L^2) = \frac{1}{K}\sum_{i=1}^K\frac{|\L^{1}_i\Delta\L^2_i|}{2i},
\label{eq:isim}
\end{equation*}
where $|\L^{1}_i\Delta\L^2_i|$ denotes the cardinality of the set $\L^{1}_i\Delta\L^2_i$, which is the symmetric difference between $\L^{1}_i$ and $\L^2_i$. 
When the ordered sequences contained in $\L^{1}$ and $\L^2$ are completely different at level $K$, then $\II_K = 0$. 
On the other hand, $\II_K = 1$ when the top $K$ entries of the two ordered ranking lists coincide. 
Thus, the higher the value of $\II_K$, the better agreement between the top $K$ rankings provided by the~two~lists (see e.g.\ \cite{FKS03}).

\begin{figure}[t]
    \centering
    \includegraphics[width=1\columnwidth,clip,trim=0 0 0 0]{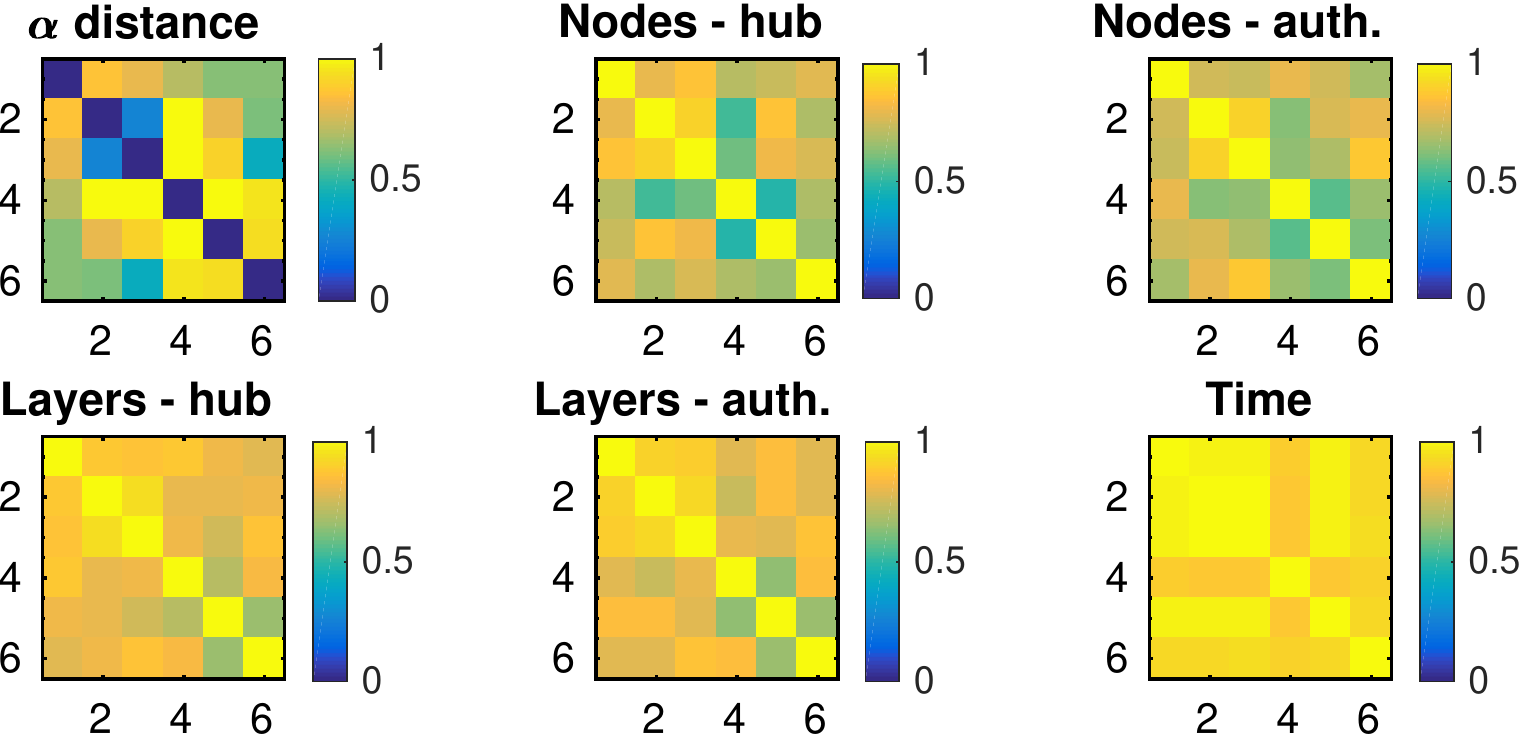}
    \caption{Top left: (Euclidean) distance matrix $D$ for different choices of $\b\alpha$. The first row/column is the distance from ${\b \alpha}^{(0)}$. Other plots: $\II_{K}$ (top center/right and bottom left/center: $K=100$; bottom right: $K=n_T$) between every pair of ranking vectors derived from MD-HITS with exponents in $\{{\b \alpha}^{(0)},\b\alpha^{(1)},\ldots,\b\alpha^{(5)}\}$. }
    \vspace{-1em}
    \label{fig:5choices-alpha}
\end{figure}

The top left plot in Fig.~\ref{fig:5choices-alpha} shows the 
distance matrix $D\in\RR^{6\times 6}$, whose entries $(D)_{st} =\|{\b \alpha}^{(s-1)}-{\b \alpha}^{(t-1)}\|_2$ for all $s,t=1,\dots,6$
are the distances between pairs of exponents in the set $\{{\b \alpha}^{(0)}, \b\alpha^{(1)},\ldots,\b\alpha^{(5)}\}$.
In the remaining five plots of Fig.~\ref{fig:5choices-alpha} we display the value of 
$\II_K$ for the different ranking vectors, for all the possible pairs of choices of $\b\alpha$ 
ordered as in the matrix $D$, with $K = 100$ for the centralities of nodes and layers and $K=n_T$ for the importance of time stamps. 
From these plots we can clearly see that the rankings are all very similar, even though the selected vectors of exponents substantially differ. In particular, the first row (and thus column) of each of these plots displays high values of $\II_K$.
This confirms that the uniform choice $\b\alpha^{(0)}$ of exponents can be preferred to others without compromising the output of the algorithm.

\begin{figure}
    \centering
    \includegraphics[width=1\columnwidth]{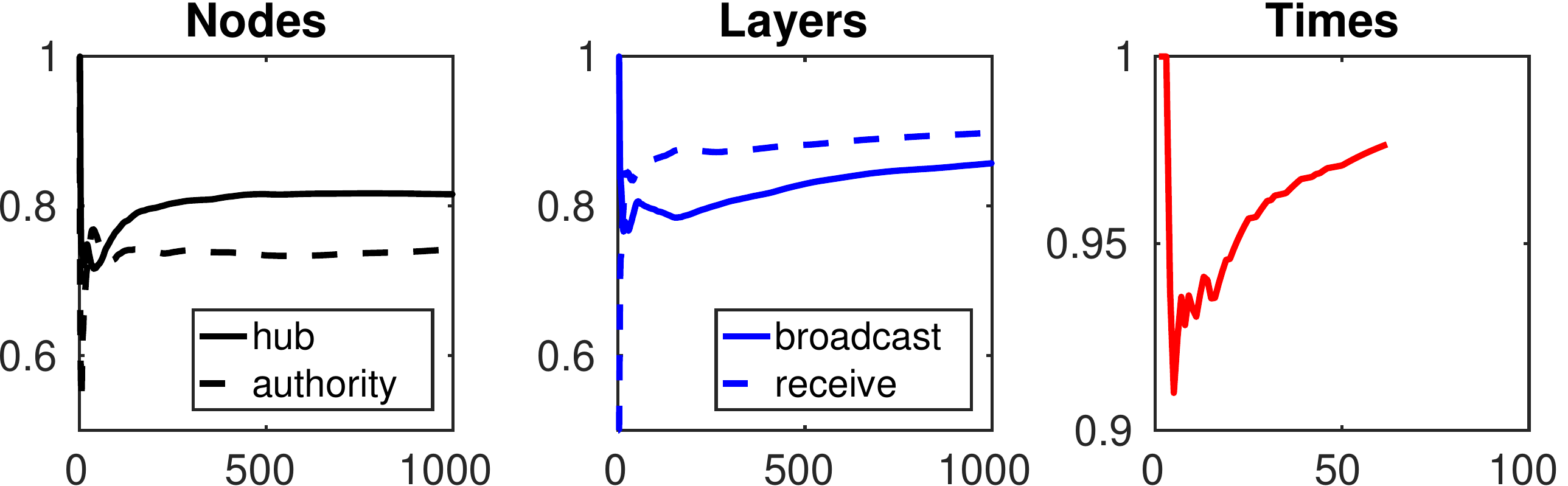}
    \caption{Evolution of $\II_K$ between the rankings derived from the centralities computed with ${\b \alpha}^{(0)}$ and the median of those computed using $\b\alpha^{(1)},\ldots,\b\alpha^{(5)}$. Left: node centrality, $K = 1,\ldots,1000$; center: layer centrality, $K = 1,\ldots, 1000$; right: time centrality, $K = 1,\ldots,n_T$. }
    \vspace{0em}
    \label{fig:isim-unif}
\end{figure}

\begin{figure*}[!h]
$\,$\hfill
    \begin{minipage}{.88\textwidth}
    \centering
    \includegraphics[width=1\textwidth,clip]{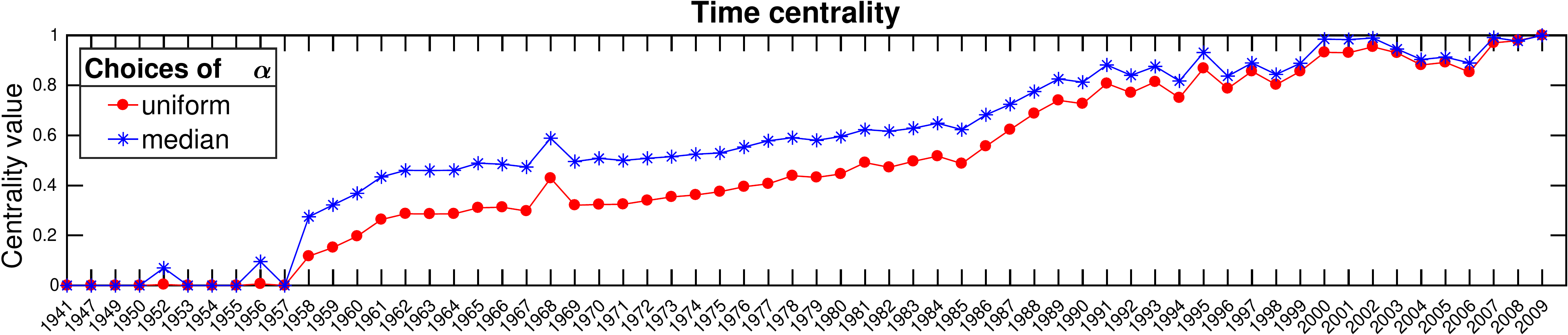}
    \end{minipage}
    \hfill
    \begin{minipage}{.11\textwidth}
    \begin{tabular}{|c|}
    \hline
    Kendall $\tau$\\
    \hline
    {\scriptsize $\alpha^{(1)}\,\,$0.9646} \\
    {\scriptsize $\alpha^{(2)}\,\,$0.9430} \\
    {\scriptsize $\alpha^{(3)}\,\,$0.9686} \\
    {\scriptsize $\alpha^{(4)}\,\,$0.9273} \\
    {\scriptsize $\alpha^{(5)}\,\,$0.9489} \\
    \hline
    \end{tabular}
    \end{minipage}
    \caption{Left: centrality scores for the time stamps in the citation dataset for $\b \alpha^{(0)} = (1,\dots,1)/5$ (red dots) and median of the centrality scores obtained with $\{\b\alpha^{(1)}, \ldots,
    \b\alpha^{(5)}\}$ (blue stars). Right: Kendall $\tau$ correlation coefficient between the years centrality for $\b \alpha^{(0)}$ and every other $\b \alpha^{(1)}, \dots, \b \alpha^{(5)}$.}
    \label{fig:my_label}
\end{figure*}

In Fig.~\ref{fig:isim-unif} we display the evolution of the measure $\II_K$  between the rankings obtained with  ${\b \alpha}^{(0)}$ and the median of the centrality vectors computed with $\{\b  \alpha^{(1)},\ldots, \b\alpha^{(5)}\}$ for $K=1,2,\ldots,n$. Here, $n=1000$ for node and layer centrality vectors and $n=n_T$ for the importance vector of the time stamps.  
A solid line depicts the evolution of $\II_K$ for $\b h$ (left), $\b b$ (center), and $\b \tau$ (right), while a dashed line displays the evolution of $\II_K$ for $\b a$ (left) and $\b r$ (center). 
Overall, the behavior of $\II_K$ further  demonstrates the  robustness of the model with respect to the choice of exponents~$\b\alpha$.

In Fig.~\ref{fig:my_label} we display in chronological order the scores of each time stamp for the vector computed with ${\b \alpha}^{(0)}$ (circle) and for the median of the vectors computed using the exponents in $\{\b  \alpha^{(1)},\ldots, \b\alpha^{(5)}\}$ (star).  
The actual centrality scores do not perfectly match in the two settings. However, the induced rankings almost coincide thus confirming the behaviour observed in Fig.~\ref{fig:isim-unif}. 
This is further supported by the very high values achieved by the Kendall $\tau$ correlation coefficients (right of Fig.~\ref{fig:my_label}) between the rankings obtained with $\b\alpha^{(0)}$ and any other choice $\b\alpha^{(i)}$.  

The \textit{quality} of the derived rankings is not easily quantifiable, as there are no objective criteria to rely on. Moreover, for this specific dataset,  domain specific knowledge would be required for an assessment. 
We can however comment on the ranking derived from our time centrality.
Fig.~\ref{fig:my_label} shows that recent years have a higher percentage of importance compared to earlier times. 
This is consistent with what one would expect as 
1) the volume of papers published per year has recently considerably increased, and 2)   research papers are far more easily accessible now than in earlier times, making it easier for researchers to cite each other.

\subsection{FAO Dataset}
We now move on to the analysis of the FAO Dataset 2010 \cite{de2015structural}. 
This static network ($n_T = 1$) contains $n_V = 214$ nodes representing nations in the world, $318346$ directed edges between the nodes, and $n_L = 364$ layers representing goods. An edge between two nodes represents an import/export relationship of a specific good between the two countries. 
There are no edges across layers. 
Every country considered in this dataset exports at least one product, that is $\sum_{j,k}\A_{ijk}\neq 0$ for any $i=1,\dots,n_V$. 
On the other hand, there are 81 countries which do not import, resulting in 81 zero unfoldings $\sum_{i,k}\A_{ijk}=0$. 

We computed the MD-HITS centrality (with ${\b \alpha}^{(0)}$) on this rather sparse, disconnected network and compared it with available techniques for multiplex graphs: aggregate degree~\cite{BNL14}, aggregate HITS~\cite{SRC13}, and eigenvector versatility~\cite{DSOGA15}.  
In Fig. \ref{fig:scatter} we display the scatter plots of the MD-HITS vectors $\b h$ (top) and $\b a$ (bottom), versus the other centrality measures. 
It can be clearly seen from these plots that the rankings provided by the available techniques differ from the ones returned by MD-HITS. 
Since the aggregate graph of this dataset consists of $82$ strongly connected components, aggregate HITS and eigenvector versatility are not well defined. 
This results in an ambiguity in the centrality vectors computed, since there is more than one possible solution (here we are displaying the one obtained using one run of Matlab's built-in function {\tt eigs}). 
Moreover, these measures assign zero score to several non-negligible nodes: eigenvector versatility, e.g., incorrectly assigns zero broadcasting score to $26$ nodes that have positive aggregate outdegree (see top-right of Fig.~\ref{fig:scatter}). 
On the other hand, MD-HITS assigns a unique and positive hub score to all the nodes; moreover, it assigns positive authority score to every node except for exactly those $81$  which correspond to countries that do not import goods.  
Finally, note that MD-HITS is also the only centrality measure that returns a ranking of the layers, allowing for a better interpretation of the results. 
As an example, we see that for all centrality measures the top ranked hub is the USA, while the top receiver is China for all measures but eigenvector versatility, that ranks Canada first. 
This result may seem strange if we do not consider the importance of layers. 
From the MD-HITS layer centrality it can be seen that the most important food product in this import/export network is soybeans, confirming that China (resp.,\ USA) is the most important receiver (resp.,\ broadcaster)  in the network, as it imported \$22.6B worth of soybeans in 2010, mainly from the USA~\cite{mit_OEC}. 
The second country exporting soybeans to China was Brazil, identified as the second best hub by MD-HITS (ranked third according to aggregate degree, sixth by aggregate HITS, and not amongst the top ten according to eigenvector versatility). 
These results, together with the fact that the other eigenvector-based measures are not well defined for this dataset, showcase the many advantages of  MD-HITS over previously proposed eigenvector-based models.


\begin{figure}[t]
\begin{minipage}{1\columnwidth}
\centering
\includegraphics[width=1\textwidth]{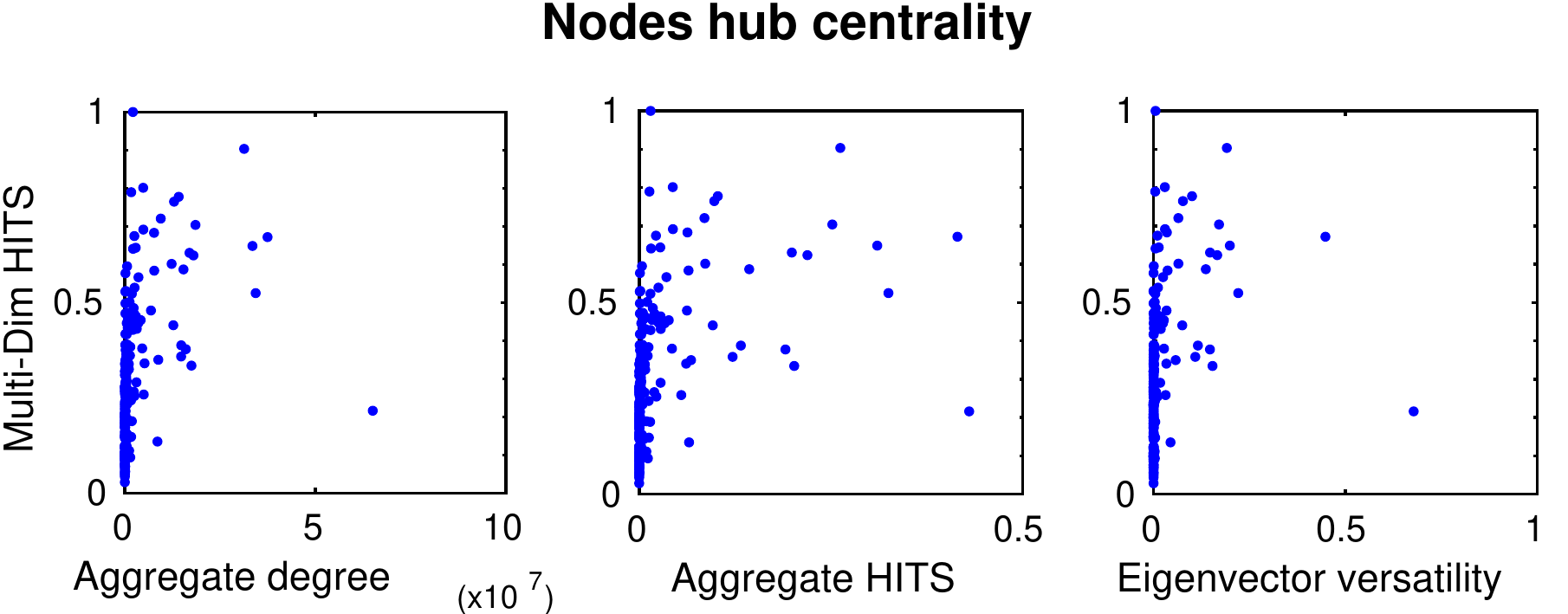}
\end{minipage}\\
\begin{minipage}{1\columnwidth}
\centering
\includegraphics[width=1\textwidth]{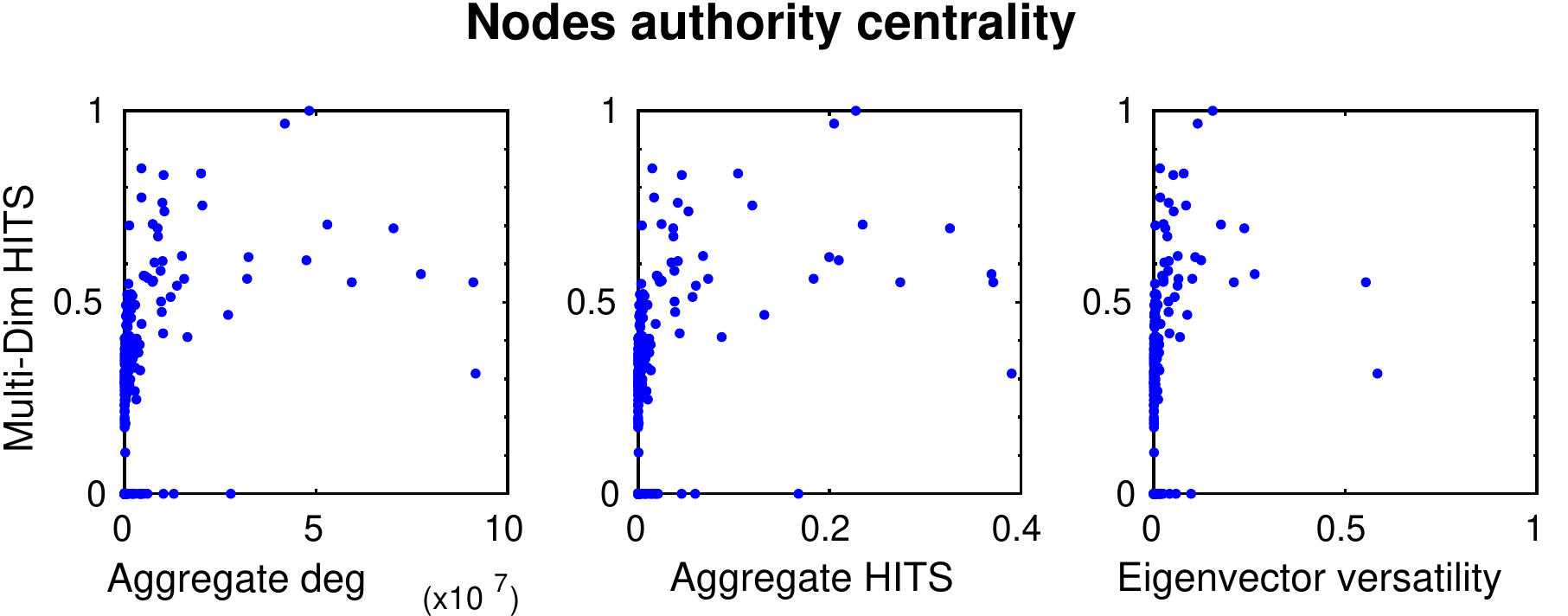}
\end{minipage}
\caption{Scatter plot of the centrality vectors for nodes computed via MD-HITS versus aggregate degree, aggregate eigenvector, and eigenvector versatility. Top: hub centrality. Bottom: authority centrality}
\label{fig:scatter}
\vspace{0em}
\end{figure}

\section{Conclusions}
We introduced a new ranking model for temporal directed multilayer networks, extending the mutually reinforcing nature of HITS algorithm to this framework. 
The new centrality vectors are always computable for nonnegative tensors and global convergence of the algorithm is always guaranteed in practical situation thanks to the introduction of nonlinearity in the model. 
Numerical experiments on real world networks demonstrate the scalability and illustrate the potential of the proposed ranking algorithm. 

\section*{Acknowledgements}
EPSRC Data Statement: the code and data used in this work is publicly available at 
{\small\tt https://github.com/ftudisco/multi-dimensional-hits}.

\end{document}